\documentclass[11pt]{article}
\usepackage[left=1in,top=1in,right=1in,bottom=1in]{geometry}
\usepackage{times}

\usepackage{url}
\usepackage{verbatim}
\usepackage{amsmath}
\usepackage{amssymb}
\usepackage{amsthm}
\usepackage{rotating}
\usepackage{algorithm}
\usepackage{algorithmic}

\usepackage{tabularx}
\usepackage{subfigure}
\usepackage{multirow}
\usepackage{siunitx}
\usepackage{makecell}
\usepackage{placeins}

\theoremstyle{definition}
\theoremstyle{plain}

\newtheorem{definition}{Definition}

\newtheorem{proposition}{Proposition}

\DeclareMathOperator*{\argmax}{arg\,max}
\DeclareMathOperator*{\argmin}{arg\,min}

\allowdisplaybreaks

\usepackage{silence}
\begin{document}

\title{Evolutionarily Stable Stackelberg Equilibrium\thanks{Accepted for presentation at the 65th IEEE Conference on Decision and Control (CDC 2026).}}

\author{Sam Ganzfried\\
Ganzfried Research, Cornell University\\
sam.ganzfried@gmail.com
}

\date{\vspace{-5ex}}

\maketitle

\begin{abstract}
We introduce evolutionarily stable Stackelberg equilibrium (SESS) for games in which a single leader interacts with a symmetric population of followers.
The leader selects an optimal strategy, anticipating that the follower population satisfies specified evolutionary stability conditions. We use classical evolutionarily stable strategies in discrete games and ecological equilibrium with nonpositive invasion fitness in the continuous model. We consider both leader-optimal and leader-pessimal selection among stable follower responses, which arise as special cases of our framework. Earlier cancer-treatment models express the follower's evolutionary response through fitness maximization at ecological equilibrium. Our framework specifies which follower outcomes are admissible and how an outcome is selected when several satisfy the stability criterion. We present algorithms for computing SESS in discrete and continuous games. Our model applies naturally to biological settings; for example, in cancer treatment the leader represents the physician and the followers correspond to competing cancer cell phenotypes. Experiments on a cancer game illustrate different treatment strategies and quality-of-life values under the previous Stackelberg equilibrium formulation and our optimistic SESS formulation.
\end{abstract}

\section{Introduction}
\label{se:intro}
While Nash equilibrium has emerged as the standard solution concept in game theory, it is often criticized as being too weak: often games contain multiple Nash equilibria (sometimes even infinitely many), and we want to select one that satisfies other natural properties. For example, one popular concept that refines Nash equilibrium is evolutionarily stable strategy (ESS). A mixed strategy in a two-player symmetric game is an evolutionarily stable strategy if it is robust to being overtaken by a mutation strategy. Formally, $\mathbf{x}^*$ is an ESS if for every mixed strategy $\mathbf{x}$ that differs from $\mathbf{x}^*$, there exists $\epsilon_0 = \epsilon_0(\mathbf{x}) > 0$ such that, for all $\epsilon \in (0,\epsilon_0)$, 
\begin{equation*}
(1-\epsilon)u_1(\mathbf{x},\mathbf{x}^*) + \epsilon u_1(\mathbf{x},\mathbf{x}) <  (1-\epsilon)u_1(\mathbf{x}^*,\mathbf{x}^*)+\epsilon u_1(\mathbf{x}^*,\mathbf{x}).
\end{equation*}
Here $\mathbf{x}^*$ is the resident strategy, $\mathbf{x}$ is the mutant
strategy, and $\epsilon$ is the fraction of mutants in the population.
The inequality says that when mutants are sufficiently rare, their expected
payoff is lower than that of residents.
ESS was initially proposed by mathematical biologists motivated by applications such as population dynamics (e.g., maintaining robustness to mutations within a population of humans or animals)~\cite{Maynard73:Logic,Maynard82:Evolution}. A common example game is the 2x2 game where strategies correspond to an ``aggressive'' Hawk or a ``peaceful'' Dove strategy. Evolutionary game models have also been applied to interactions between normal and malignant cells; for example, Dingli et al. study interactions among osteoblasts, osteoclasts, and multiple myeloma cells, with implications for cancer treatment~\cite{Dingli09:Cancer}. 

While Nash equilibrium is defined for general multiplayer games, ESS is traditionally defined specifically for two-player symmetric games. ESS is a refinement of Nash equilibrium. In particular, if $\mathbf{x}^*$ is an ESS, then $(\mathbf{x}^*,\mathbf{x}^*)$ (i.e., the strategy profile where both players play $\mathbf{x}^*$) is a (symmetric) Nash equilibrium~\cite{Maschler13:Game}. Of course the converse is not necessarily true (not every symmetric Nash equilibrium is an ESS), or else ESS would be a trivial refinement. In fact, ESS is not guaranteed to exist in finite two-player symmetric games (while Nash equilibrium is guaranteed to exist in all finite games). For example, while rock-paper-scissors has a (symmetric) mixed strategy Nash equilibrium (which puts equal weight on all three actions), it has no ESS~\cite{Maschler13:Game}.

There exists a polynomial-time algorithm for computing Nash equilibrium (NE) in two-player zero-sum games, while computing an NE in a general two-player game is PPAD-complete~\cite{Chen06:Settling} and it is widely conjectured that no efficient (polynomial-time) algorithm exists. However, several algorithms have been devised that perform well in practice. The problem of computing whether a game has an ESS was shown to be both NP-hard and coNP-hard and also to be contained in $\Sigma ^P _2$ (the class of decision problems that can be solved in nondeterministic polynomial time given access to an NP oracle)~\cite{Etessami08:Computational}. Subsequently it was shown that the exact complexity of this problem is that it is $\Sigma^P _2$-complete~\cite{Conitzer13:Exact}. Note that this result is for determining whether an ESS exists (as discussed above there exist games which have no ESS), not for the complexity of computing an ESS in games for which one exists. Several approaches have been proposed for computing ESS in two-player games~\cite{Haigh75:Game,Abakuks80:Conditions,Broom13:Game-Theoretical,Bomze92:Detecting,Mcnamara97:General}.

We introduce \emph{evolutionarily stable Stackelberg equilibrium} (SESS) to model a rational leader interacting with an evolving follower population. Cancer treatment provides a natural motivation. The physician selects a treatment to maximize quality of life, while the cancer cells form an evolving population subject to selection. We model the physician as anticipating a follower outcome that is resistant to invasion by rare mutants. Best-response conditions can also describe outcomes of evolutionary processes, but in discrete games they do not by themselves ensure ESS. SESS makes the required stability criterion and follower selection rule explicit. Prior Stackelberg evolutionary game models couple ecological equilibrium and fitness-maximizing trait responses with the leader's optimization problem~\cite{Kleshnina23:Game,Stein23:Stackelberg}. Our contribution is the SESS solution-concept framework, with explicit admissibility and selection rules for stable follower outcomes, optimistic and pessimistic variants, and discrete and continuous computational formulations.

\section{Discrete Stackelberg--ESS}
\label{se:discrete-sess}
A \emph{normal-form game} consists of a finite set of players $N = \{1,\ldots,n\}$, a finite set of pure strategies $S_i$ for each player $i$, and a real-valued utility for each player for each strategy vector (aka \emph{strategy profile}), $u_i : \times_i S_i \rightarrow \mathbb{R}$. In a \emph{symmetric normal-form game}, all players have the same pure strategy set $S_i=S$ and the utility functions satisfy the following symmetry condition: for every player $i \in N$, pure strategy profile $(s_1,\ldots,s_n) \in S^n$, and permutation $\pi$ of the players,
$$u_i(s_1,\ldots,s_n) = u_{\pi(i)}(s_{\pi^{-1}(1)},\ldots,s_{\pi^{-1}(n)}).$$
This allows us to remove the player index of the utility function and just write 
$u(s_1,\ldots,s_n),$ where it is implied that the utility is for player 1 (we can simply permute the players to obtain the utilities of the other players). We write $u_i$ for notational convenience, but note that only a single utility function must be specified which applies to all players.

Let $\Sigma_i$ denote the set of mixed strategies of player $i$ (probability distributions over elements of $S_i$). If players follow mixed strategy profile
$\mathbf{x} = (\mathbf{x}^{(1)},\ldots,\mathbf{x}^{(n)}),$ where $\mathbf{x}^{(i)} \in \Sigma_i$, expected payoff to player $i$ is
$$u_i(\mathbf{x}^{(1)},\ldots,\mathbf{x}^{(n)}) = \sum_{s_1,\ldots,s_n \in S} x^{(1)}_{s_1} \cdots x^{(n)}_{s_n}u_i(s_1,\ldots,s_n).$$ 
We write $u_i(\mathbf{x}) = u_i(\mathbf{x}^{(i)},\mathbf{x}^{(-i)})$,
where $\mathbf{x}^{(-i)}$ denotes the vector of strategies of all players except $i$.
If all players follow the same mixed strategy $\mathbf{x}$, then for all players $i$:
$$u_i(\mathbf{x}) = u_i(\mathbf{x},\ldots,\mathbf{x}) = \sum_{s_1,\ldots,s_n \in S} x_{s_1} \cdots x_{s_n} u_i(s_1,\ldots,s_n).$$ 

\begin{definition}
\label{de:ne}
A mixed strategy profile $\mathbf{x}^\star$ is a Nash equilibrium if for each player $i \in N$ and for each mixed strategy $\mathbf{x}^{(i)} \in \Sigma_i$:
$u_i(\mathbf{x}^{*(i)},\mathbf{x}^{*(-i)}) \geq u_i(\mathbf{x}^{(i)},\mathbf{x}^{*(-i)}).$
\end{definition}

\begin{definition}
\label{de:ne-sym}
A mixed strategy profile $\mathbf{x}^\star$ in a symmetric normal-form game is a symmetric Nash equilibrium if it is a Nash equilibrium and:
$\mathbf{x}^{*(1)} = \mathbf{x}^{*(2)} = \cdots = \mathbf{x}^{*(n)}.$
\end{definition}

\begin{definition}
\label{de:ess}
A mixed strategy $\mathbf{x}^\star \in \Sigma_1$ is evolutionarily stable in a two-player symmetric normal-form game if for each mixed strategy $\mathbf{x} \neq \mathbf{x}^\star$ exactly one of the following holds:
\begin{enumerate}
    \item $u_1(\mathbf{x}^*,\mathbf{x}^*) > u_1(\mathbf{x},\mathbf{x}^*),$
    \item $u_1(\mathbf{x}^*,\mathbf{x}^*) = u_1(\mathbf{x},\mathbf{x}^*)$ and $u_1(\mathbf{x}^*,\mathbf{x}) > u_1(\mathbf{x},\mathbf{x}).$
\end{enumerate}
\end{definition}

It has been proven that every symmetric normal-form game has at least one symmetric Nash equilibrium~\cite{Nash51:Non}.
It is clear from Definition~\ref{de:ess} that every ESS in a two-player symmetric normal-form game must be a symmetric Nash equilibrium (SNE).
Thus, a natural approach for ESS computation is to first compute SNE and then perform subsequent procedures to determine whether they are ESS. 

Now we present the Stackelberg evolutionary game setting. The leader's set of pure strategies is $\mathcal{M} = \{1,\ldots,m\}$ and the follower phenotypes are $\mathcal{P} = \{1,\ldots,n\}.$ The leader's utility function is given by $\mathbf{A}_L \in \mathbb{R}^{m \times n}$, where $\mathbf{A}_L(\ell,i)$ gives the payoff to the leader when the leader plays pure strategy $\ell \in \mathcal{M}$ and the follower population state is phenotype $i \in \mathcal{P}.$ The leader's expected utility when the leader plays mixed strategy $\boldsymbol\sigma$ and the follower plays mixed strategy $\mathbf{x}$ is $U_L(\boldsymbol\sigma,\mathbf{x}) = \sum_{\ell = 1}^m \sum_{i=1}^n \sigma_{\ell} x_i \mathbf{A}_L(\ell,i).$ The follower payoff tensor is $\mathbf{A}_F \in \mathbb{R}^{m \times n \times n},$ where $\mathbf{A}_F(\ell, i, j)$ is the payoff (fitness) to a follower of phenotype $i \in \mathcal{P}$ when the leader plays pure strategy $\ell \in \mathcal{M}$ and the interacting follower has phenotype $j \in \mathcal{P}.$ For each fixed $\ell \in \mathcal M$, $\mathbf{A}_F(\ell,\cdot,\cdot)$ is the $n \times n$ payoff matrix of the induced symmetric follower game. After the leader commits to mixed strategy $\boldsymbol\sigma \in \Delta(\mathcal{M})$, the induced follower payoff matrix is $\mathbf{B}^{\boldsymbol\sigma} \in \mathbb{R}^{n \times n}$ with $\mathbf{B}^{\boldsymbol\sigma} = \sum_{\ell = 1}^m \sigma_{\ell} \mathbf{A}_F(\ell,\cdot,\cdot).$ The follower population then plays a symmetric evolutionary game with payoff matrix $\mathbf{B}^{\boldsymbol\sigma}.$ Let $\mathcal N(\boldsymbol\sigma)$ denote the set of symmetric Nash equilibria of $\mathbf{B}^{\boldsymbol\sigma}.$ Given a selection correspondence $\mathcal H(\boldsymbol\sigma) \subseteq \mathcal N(\boldsymbol\sigma)$, a pair $(\boldsymbol\sigma^*,\mathbf{x}^*)$ is a \emph{Stackelberg equilibrium} (SE) if
\[
\begin{aligned}
\boldsymbol\sigma^*&\in\argmax_{\boldsymbol\sigma\in\Delta(\mathcal M)}
\max_{\mathbf x\in\mathcal H(\boldsymbol\sigma)}U_L(\boldsymbol\sigma,\mathbf x),\\
\mathbf x^*&\in\argmax_{\mathbf x\in\mathcal H(\boldsymbol\sigma^*)}
U_L(\boldsymbol\sigma^*,\mathbf x).
\end{aligned}
\]
A population state $\mathbf{x} \in  \Delta(\mathcal{P})$ is admissible for $\boldsymbol\sigma$ iff it is an ESS of $\mathbf{B}^{\boldsymbol\sigma}.$ Let $\mathcal E(\boldsymbol\sigma)$ denote the set of ESSs of $\mathbf{B}^{\boldsymbol\sigma}.$ We assume that $\mathcal E(\boldsymbol\sigma) \neq \emptyset$ for all $\boldsymbol\sigma \in \Delta(\mathcal M),$ so that the follower response correspondence is well-defined. Let \(g^{\boldsymbol\sigma} : \Delta(\mathcal P) \to \mathbb R\) denote an \emph{ESS selection function}. Finally let $\mathcal G(\boldsymbol\sigma) = \argmax_{\mathbf{x} \in \mathcal E(\boldsymbol\sigma)} g^{\boldsymbol\sigma}(\mathbf{x}).$ The leader maximizes utility anticipating selection from $\mathcal G(\boldsymbol\sigma)$. The inner maximum specifies optimistic selection among any remaining ties, so the selected outcome must attain that maximum. For the equilibrium definitions, we assume that the selection correspondences are nonempty and the displayed leader and follower extrema are attained. These assumptions do not assert general existence. In particular, the ESS response correspondence need not have a closed graph, so compact strategy spaces alone do not guarantee attainment of the leader optimum.

We now present our general definition for \emph{evolutionarily stable Stackelberg equilibrium} (SESS) in Definition~\ref{def:SESS}. We also define two natural special cases. In the first case, the followers play the ESS that is most beneficial to the leader. We refer to this as an \emph{optimistic evolutionarily stable Stackelberg equilibrium} (OSESS). In this case we have $g^{\boldsymbol\sigma}(\mathbf{x}) = U_L(\boldsymbol\sigma, \mathbf{x}).$ We also consider \emph{pessimistic evolutionarily stable Stackelberg equilibrium} (PSESS), in which the followers play the ESS that is worst for the leader. Pessimistic selection allows the leader to account for the least favorable ESS without assuming that the followers intend to harm the leader.
For PSESS we have $g^{\boldsymbol\sigma}(\mathbf{x}) = -U_L(\boldsymbol\sigma, \mathbf{x}).$ OSESS and PSESS are defined below in Definitions~\ref{def:OSESS}--\ref{def:PSESS}. For OSESS we have $\mathcal G(\boldsymbol\sigma) = \argmax_{\mathbf{x} \in \mathcal E(\boldsymbol\sigma)} U_L(\boldsymbol\sigma, \mathbf{x})$, and for PSESS $\mathcal G(\boldsymbol\sigma) = \argmin_{\mathbf{x} \in \mathcal E(\boldsymbol\sigma)} U_L(\boldsymbol\sigma, \mathbf{x}).$ It is clear that OSESS and PSESS are special cases of SESS (since they specify a particular function $g^{\boldsymbol\sigma}(\mathbf{x})$). Requiring the follower response to be an ESS refines the set of admissible follower responses in the induced population game. This restriction can change both the selected follower response and the leader's optimal commitment. Proposition~\ref{pr:SESS-discrete} shows that every SESS is a Stackelberg equilibrium under the same selection correspondence.

\begin{definition}[Evolutionarily Stable Stackelberg Equilibrium]
\label{def:SESS}
Given a selection correspondence $\mathcal G(\boldsymbol\sigma) \subseteq \mathcal E(\boldsymbol\sigma)$, $(\boldsymbol\sigma^*,\mathbf{x}^*)$ is an
\emph{evolutionarily stable Stackelberg equilibrium} if
\[
\begin{aligned}
\boldsymbol\sigma^*&\in\argmax_{\boldsymbol\sigma\in\Delta(\mathcal M)}
\max_{\mathbf x\in\mathcal G(\boldsymbol\sigma)}U_L(\boldsymbol\sigma,\mathbf x),\\
\mathbf x^*&\in\argmax_{\mathbf x\in\mathcal G(\boldsymbol\sigma^*)}
U_L(\boldsymbol\sigma^*,\mathbf x).
\end{aligned}
\]
\end{definition}

\begin{definition}[Optimistic SESS]
\label{def:OSESS}
A pair $(\boldsymbol\sigma^*, \mathbf{x}^*) \in \Delta(\mathcal{M}) \times \Delta(\mathcal{P})$ is an
\emph{optimistic evolutionarily stable Stackelberg equilibrium} if
\[
(\boldsymbol\sigma^*,\mathbf{x}^*) \in 
\argmax_{\boldsymbol\sigma\in\Delta(\mathcal M),\, \mathbf{x}\in\mathcal E(\boldsymbol\sigma)}
U_L(\boldsymbol\sigma,\mathbf{x}).
\]
\end{definition}

\begin{definition}[Pessimistic SESS]
\label{def:PSESS}
A pair $(\boldsymbol\sigma^*, \mathbf{x}^*) \in \Delta(\mathcal{M})\times \Delta(\mathcal{P})$ is a
\emph{pessimistic evolutionarily stable Stackelberg equilibrium} if
\[
\boldsymbol\sigma^* \in \argmax_{\boldsymbol\sigma\in\Delta(\mathcal M)}
\;\min_{\mathbf{x}\in\mathcal E(\boldsymbol\sigma)}
U_L(\boldsymbol\sigma,\mathbf{x})
\text{ and }
\mathbf{x}^* \in 
\argmin_{\mathbf{x} \in \mathcal E(\boldsymbol\sigma^*)} U_L(\boldsymbol\sigma^*,\mathbf{x}).
\]
\end{definition}

\begin{proposition}
Every evolutionarily stable Stackelberg equilibrium $(\boldsymbol\sigma^*,\mathbf{x}^*)$ under selection correspondence $\mathcal G$ is a Stackelberg equilibrium under the same selection correspondence $\mathcal G$.
\label{pr:SESS-discrete}
\end{proposition}

\begin{proof}
Suppose that $(\boldsymbol\sigma^*,\mathbf{x}^*)$ is an evolutionarily stable Stackelberg equilibrium under selection correspondence $\mathcal G.$ By Definition~\ref{def:SESS},
\[
\boldsymbol\sigma^* \in \argmax_{\boldsymbol\sigma\in\Delta(\mathcal{M})}
\;\max_{\mathbf{x}\in \mathcal G(\boldsymbol\sigma)}
U_L(\boldsymbol\sigma,\mathbf{x})
\]
and
\[
\mathbf{x}^* \in \argmax_{\mathbf x\in\mathcal G(\boldsymbol\sigma^*)}U_L(\boldsymbol\sigma^*,\mathbf x).
\]
Since every ESS is a symmetric Nash equilibrium,
\[
\mathcal G(\boldsymbol\sigma)
\subseteq \mathcal E(\boldsymbol\sigma)
\subseteq \mathcal N(\boldsymbol\sigma)
\]
for every $\boldsymbol\sigma.$ Thus, $\mathcal G$ is also a valid Stackelberg selection correspondence. The two conditions above are exactly the conditions for a Stackelberg equilibrium under selection correspondence $\mathcal G.$
\end{proof}

Restricting follower responses from symmetric Nash equilibria to ESSs can strictly reduce the set of admissible responses. Even under optimistic selection in both models, this restriction can change which leader--follower pairs are optimal. In particular, there can exist a Stackelberg equilibrium whose selected follower response is a symmetric Nash equilibrium but not an ESS. Consider a Stackelberg evolutionary game with a single leader action and two follower phenotypes. The leader payoff is
\[
\mathbf{A}_L =
\begin{bmatrix}
1 & 0
\end{bmatrix},
\]
and the induced follower payoff matrix is
\[
\mathbf{B} =
\begin{bmatrix}
0 & 0 \\
0 & 1
\end{bmatrix}.
\]
The follower population state \(\mathbf{e}_1\) is a symmetric Nash equilibrium, since both follower phenotypes obtain payoff \(0\) against \(\mathbf{e}_1\). However, \(\mathbf{e}_1\) is not an ESS: the mutant \(\mathbf{e}_2\) ties against the resident \(\mathbf{e}_1\), since
\[
\mathbf{e}_2^\top \mathbf{B}\mathbf{e}_1
=
\mathbf{e}_1^\top \mathbf{B}\mathbf{e}_1
=
0,
\]
but does better against itself than the resident does against it,
\[
\mathbf{e}_2^\top \mathbf{B}\mathbf{e}_2
=
1
>
0
=
\mathbf{e}_1^\top \mathbf{B}\mathbf{e}_2.
\]
By contrast, \(\mathbf{e}_2\) is an ESS because it is a strict symmetric Nash equilibrium. Since the leader has only one action, the leader's strategy is fixed. However, under optimistic Stackelberg selection over symmetric Nash equilibria, the selected follower response is \(\mathbf{e}_1\), because it gives the leader payoff \(1\), while \(\mathbf{e}_2\) gives the leader payoff \(0\). Under optimistic SESS, the only admissible ESS response is \(\mathbf{e}_2\), giving leader payoff \(0\). Thus an optimistic Stackelberg equilibrium need not be an optimistic SESS. This example establishes a distinction between optimistic selection over symmetric Nash equilibria and optimistic selection over ESSs under the same follower game. The cancer-treatment comparison below additionally changes the ecological feasible set, so its payoff difference should not be interpreted as isolating the effect of evolutionary stability.

\section{Continuous Stackelberg--ESS}
We now extend the Stackelberg--ESS framework to continuous Stackelberg evolutionary games,
which arise naturally in biological and control applications such as cancer treatment.
In this setting, the leader selects a continuous decision variable, while the followers
correspond to an evolving population whose response is described by coupled ecological
and evolutionary dynamics.

Let $\mathcal M \subseteq \mathbb R^d_+$ denote the leader decision set (e.g., treatment doses or control parameters). For a given leader decision $\mathbf{m} \in \mathcal M,$ the followers are described by a collection of continuous state variables, including population sizes and traits. We denote the follower outcome by $\mathbf{z}=(\mathbf{x},\mathbf{u})\in\mathcal Z,$ where $\mathbf{x}\in\mathbb R^n_+$ represents population abundances and $\mathbf{u}\in U\subseteq \mathbb R^n$ represents trait or strategy variables. As prior work has done we will assume that $U=[0,1]^n.$ In the context of cancer treatment $d$ denotes the number of drugs and $n$ denotes the number of cancer cell phenotypes.
The leader objective is given by a function
$Q : \mathcal{M} \times \mathcal{Z} \to \mathbb{R}$.
For each $\mathbf{m} \in \mathcal{M}$, the leader induces an eco--evolutionary system governing
the dynamics of $(\mathbf{x},\mathbf{u})$.
Let $\mathcal N(\mathbf{m}) \subseteq \mathcal Z$ denote the set of admissible follower responses under the standard continuous Stackelberg model. 
Given a selection correspondence $\mathcal H(\mathbf{m}) \subseteq \mathcal N(\mathbf{m})$, a pair $(\mathbf{m}^*,\mathbf{z}^*)$ is a continuous Stackelberg equilibrium if
\[
\begin{aligned}
\mathbf{m}^*&\in\argmax_{\mathbf{m}\in\mathcal M}\max_{\mathbf z\in\mathcal H(\mathbf{m})}Q(\mathbf{m},\mathbf z),\\
\mathbf z^*&\in\argmax_{\mathbf z\in\mathcal H(\mathbf{m}^*)}Q(\mathbf{m}^*,\mathbf z).
\end{aligned}
\]

For a fixed leader decision $\mathbf{m}$, we define the set of admissible follower responses
$\mathcal{E}(\mathbf{m}) \subseteq \mathcal{Z}$ as the set of \emph{eco--evolutionarily stable}
outcomes under $\mathbf{m}$. Informally, a follower outcome $\mathbf{z}^*=(\mathbf{x}^*,\mathbf{u}^*)$ is admissible if it satisfies any required ecological equilibrium conditions and is stable against invasion by rare mutants with alternative trait values in the induced eco--evolutionary system.
When multiple evolutionarily stable outcomes exist for the same leader decision, we model their selection via a correspondence $\mathcal G(\mathbf{m}) \subseteq \mathcal E(\mathbf{m}).$ We assume that
\[
\emptyset \neq \mathcal E(\mathbf{m}) \subseteq \mathcal N(\mathbf{m})
\]
for all $\mathbf{m} \in \mathcal M,$ so that the follower response correspondence is well-defined and every eco--evolutionarily stable outcome is an admissible follower response under the standard Stackelberg model. As in the discrete setting, we assume that the selection correspondences are nonempty and the displayed extrema are attained.

We now define continuous SESS and its optimistic and pessimistic special cases. Under a common ecological model, requiring evolutionary stability restricts the set of admissible follower responses. Even under optimistic selection in both models, this restriction can change the selected follower outcome and the leader's optimal decision. Proposition~\ref{pr:SESS-continuous} shows that every continuous SESS is a continuous Stackelberg equilibrium under the same selection correspondence.

\begin{definition}[Continuous Stackelberg--ESS]
\label{def:SESS-cont}
Let $\mathcal{G}(\mathbf{m}) \subseteq \mathcal{E}(\mathbf{m})$ be a selection correspondence.
A pair $(\mathbf{m}^*, \mathbf{z}^*) \in \mathcal{M} \times \mathcal{Z}$ is a
\emph{continuous Stackelberg--ESS} if
\[
\begin{aligned}
\mathbf m^*&\in\argmax_{\mathbf m\in\mathcal M}
\max_{\mathbf z\in\mathcal G(\mathbf m)}Q(\mathbf m,\mathbf z),\\
\mathbf z^*&\in\argmax_{\mathbf z\in\mathcal G(\mathbf m^*)}Q(\mathbf m^*,\mathbf z).
\end{aligned}
\]
\end{definition}

\begin{definition}[Optimistic continuous Stackelberg--ESS]
A pair $(\mathbf{m}^*, \mathbf{z}^*) \in \mathcal{M} \times \mathcal{Z}$ is an
\emph{optimistic continuous Stackelberg--ESS} if
\[
(\mathbf{m}^*, \mathbf{z}^*) \in \argmax_{\mathbf{m} \in \mathcal{M}, \ \mathbf{z} \in \mathcal{E}(\mathbf{m})} Q(\mathbf{m},\mathbf{z}).
\]
\end{definition}

\begin{definition}[Pessimistic continuous Stackelberg--ESS]
A pair $(\mathbf{m}^*, \mathbf{z}^*) \in \mathcal{M} \times \mathcal{Z}$ is a
\emph{pessimistic continuous Stackelberg--ESS} if
\[
\mathbf{m}^* \in \argmax_{\mathbf{m} \in \mathcal{M}} \ \min_{\mathbf{z} \in \mathcal{E}(\mathbf{m})} Q(\mathbf{m},\mathbf{z})
\text{ and }
\mathbf{z}^* \in \argmin_{\mathbf{z} \in \mathcal{E}(\mathbf{m}^*)} Q(\mathbf{m}^*,\mathbf{z}).
\]
\end{definition}

\begin{proposition}
Every continuous evolutionarily stable Stackelberg equilibrium $(\mathbf{m}^*,\mathbf{z}^*)$ under selection correspondence $\mathcal G$ is a continuous Stackelberg equilibrium under the same selection correspondence $\mathcal G$.
\label{pr:SESS-continuous}
\end{proposition}

\begin{proof}
Suppose that $(\mathbf{m}^*,\mathbf{z}^*)$ is a continuous evolutionarily stable Stackelberg equilibrium under selection correspondence $\mathcal G.$ By Definition~\ref{def:SESS-cont},
\[
\mathbf{m}^* \in \argmax_{\mathbf{m}\in\mathcal M}
\;\max_{\mathbf{z}\in\mathcal G(\mathbf{m})}
Q(\mathbf{m},\mathbf{z})
\]
and
\[
\mathbf{z}^* \in\argmax_{\mathbf z\in\mathcal G(\mathbf m^*)}Q(\mathbf m^*,\mathbf z).
\]
Since
\[
\mathcal G(\mathbf{m}) \subseteq \mathcal E(\mathbf{m}) \subseteq \mathcal N(\mathbf{m})
\]
for every $\mathbf{m},$ $\mathcal G$ is also a valid continuous Stackelberg selection correspondence. The two conditions above are exactly the conditions for a continuous Stackelberg equilibrium under selection correspondence $\mathcal G.$
\end{proof}

\begingroup
\emergencystretch=1em
We consider the cancer treatment model based on prior work~\cite{Kleshnina23:Game}. The leader chooses a treatment $\mathbf m$ to maximize the quality-of-life function $Q$. The population dynamics satisfy $\dot x_i=x_iG_i$, where $G_i$ is the per-capita growth rate of phenotype $i$. In the benchmark studied in Section~\ref{se:exp}, type~0 has no evolvable trait, while types~1 and~2 have resistance traits $u_1$ and $u_2$. For type~0, the no-invasion condition in Equation~\ref{eq:continuous-osess-reform} reduces to $G_0(\mathbf m,\mathbf x)\le 0$. Equation~\ref{eq:stack} gives the Stackelberg equilibrium (SE) formulation implemented in~\cite{Ganzfried24:Computing} and used as our experimental baseline. Its population equations set the per-capita growth rate of each phenotype to zero, and each resistance trait maximizes its phenotype's fitness.
\par\endgroup

\begin{equation}
\label{eq:stack}
\begin{aligned}
\max_{\mathbf m,\mathbf u,\mathbf x}\quad & Q(\mathbf m,\mathbf u,\mathbf x) \\
\text{s.t.}\quad
& G_0(\mathbf m,\mathbf x)=0, \\
& G_i(u_i,\mathbf m,\mathbf x)=0, \quad i=1,2, \\
& u_i\in\argmax_{v\in[0,1]}G_i(v,\mathbf m,\mathbf x), \quad i=1,2, \\
& \mathbf m\in\mathcal M,\quad \mathbf x\in\mathbb R_+^3,\quad \mathbf u\in[0,1]^2.
\end{aligned}
\end{equation}

When $x_i>0$, ecological stationarity requires $G_i=0$. When $x_i=0$, the population is stationary for any value of $G_i$, so the zero-growth equations in Equation~\ref{eq:stack} impose an additional restriction at the boundary.

Equation~\ref{eq:continuous-osess} gives the optimistic continuous Stackelberg--ESS formulation. It requires ecological stationarity and nonpositive invasion fitness for every phenotype, including absent phenotypes. An absent phenotype can therefore have strictly negative growth for every possible invading trait. Optimistic selection allows us to maximize $Q$ jointly over the leader's decision and the admissible follower outcome.

\begin{equation}
\label{eq:continuous-osess}
\begin{aligned}
\max_{\mathbf{m}^*,\mathbf{u}^*,\mathbf{x}^*} \quad & Q(\mathbf{m}^*,\mathbf{u}^*,\mathbf{x}^*) \\
\text{s.t.}\quad
& \dot x_i^* = 0, \quad i = 1,\ldots,n, \\
& (\mathbf{x}^*,\mathbf{u}^*) \in \mathcal{E}(\mathbf{m}^*), \\
& \mathbf{m}^\ast \in \mathcal M, \mathbf{x}^\ast \ge \mathbf{0}, \mathbf{0} \le \mathbf{u}^\ast \le \mathbf{1}.
\end{aligned}
\end{equation}

Equation~\ref{eq:continuous-osess-reform} provides an equivalent reformulation of Equation~\ref{eq:continuous-osess} under the eco-evolutionary stability notion adopted here. Specifically, we require the resident state to be an ecological equilibrium and require nonpositive per capita growth for every mutant trait at vanishing mutant density. For the continuous-trait setting, we formulate evolutionary stability in terms of resistance to invasion by rare mutants with alternative trait values~\cite{Hofbauer98:Evolutionary,Sandholm10:Population,Geritz98:Evolutionarily}, rather than using the discrete ESS definition presented previously. Given fixed $\mathbf{m},$ we define a resident state
$(\mathbf{x}^*,\mathbf{u}^*)$ to be eco--evolutionarily stable if
$\mathbf{x}^*$ is an ecological equilibrium
(i.e., $\dot{\mathbf{x}}^*=\mathbf{0}$) and every mutant trait has nonpositive invasion fitness in the resident environment, i.e.,
\[
\sup_{u_i \in [0,1]} G_i(u_i,\mathbf{m},\mathbf{x}^*) \leq 0
\quad \forall i.
\]
Since $\dot x_i=x_iG_i,$ the ecological equilibrium constraint $\dot{x}^{*}_i=0$ is equivalent to
\[
x_i^*G_i(u_i^*,\mathbf{m}^*,\mathbf{x}^*)=0.
\]
The no-invasion condition can be written as
\[
G_i(u_i,\mathbf{m}^*,\mathbf{x}^*)\le 0
\quad \forall u_i\in[0,1],\quad i=1,\ldots,n.
\]
Substituting these two conditions into Equation~\ref{eq:continuous-osess} gives the reformulation in Equation~\ref{eq:continuous-osess-reform}. Here ``eco--evolutionarily stable'' refers specifically to ecological stationarity and nonpositive invasion fitness. This first-order criterion allows neutral mutants and does not establish asymptotic stability of the ecological dynamics or convergence of trait evolution. It is therefore distinct from the strict finite-game ESS condition in Definition~\ref{de:ess}; neutral invasion fitness can require further analysis.

\begin{equation}
\label{eq:continuous-osess-reform}
\begin{aligned}
\max_{\mathbf{m}^*,\mathbf{u}^*,\mathbf{x}^*} \quad & Q(\mathbf{m}^*,\mathbf{u}^*,\mathbf{x}^*) \\
\text{s.t.}\quad
& x_i^* \, G_i(u_i^*, \mathbf{m}^*, \mathbf{x}^*) = 0, \quad i = 1,\dots,n, \\
& G_i(u_i, \mathbf{m}^*, \mathbf{x}^*) \le 0 
\quad \forall u_i \in [0,1],\; i = 1,\dots,n, \\
& \mathbf{m}^\ast \in \mathcal M, \mathbf{x}^\ast \ge \mathbf{0}, \mathbf{0} \le \mathbf{u}^\ast \le \mathbf{1}.
\end{aligned}
\end{equation}

For the cancer benchmark, every feasible outcome of Equation~\ref{eq:stack} satisfies Equation~\ref{eq:continuous-osess-reform}. For types~1 and~2, the resident trait maximizes fitness and attains growth zero, so every alternative trait has nonpositive growth. Type~0 also has growth zero. Conversely, when all phenotypes are present, Equation~\ref{eq:continuous-osess-reform} implies Equation~\ref{eq:stack}. Ecological stationarity gives zero resident growth, and the invasion inequalities make each resident trait a global maximizer. The formulations can differ when a phenotype is absent. Consequently, the optimal value of Equation~\ref{eq:continuous-osess-reform} can exceed that of the historical SE benchmark in Equation~\ref{eq:stack}. Proposition~\ref{pr:SESS-continuous} assumes a common ecological model, whereas the two formulations compared in Table~\ref{ta:results} impose different ecological constraints when a phenotype is absent.

\section{Algorithms}
\label{se:alg}
We first present an algorithm for computing OSESS in normal-form games, followed by an algorithm for continuous-trait games. Pseudocode for the main algorithm in the normal-form setting is given in Algorithm~\ref{alg:discrete-osess}. Algorithm~\ref{alg:discrete-osess} enumerates possible follower supports.
For each support, Algorithm~\ref{alg:discrete-stack} computes a candidate
that maximizes the leader's payoff subject to the current constraints, and
Algorithm~\ref{alg:ess-check} checks whether the candidate follower strategy
passes the stability test. A violating direction yields a new constraint
for Algorithm~\ref{alg:discrete-stack}.
In this setting, we optimize over a subset of ESS outcomes defined by positive numerical margins. These margins preserve the strict inequality in the ESS condition while restricting the set of admissible outcomes.

For a symmetric Nash equilibrium $\mathbf x$ with support $T$, let $v=\mathbf x^\top\mathbf B\mathbf x$ and $J=\{j:(\mathbf B\mathbf x)_j=v\}$. Since $\mathbf x$ is a symmetric Nash equilibrium, $T\subseteq J$, and $J$ may also contain unused pure best responses. Define
\[
\mathcal D(T,J)=\left\{\mathbf d:\ \mathbf 1^\top\mathbf d=0,\
\|\mathbf d\|_2^2=1,\ d_j\geq0\ (j\notin T),\
d_j=0\ (j\notin J)\right\}.
\]
Each $\mathbf{d}\in\mathcal{D}(T,J)$ gives a direction in which the resident strategy $\mathbf{x}$ can be perturbed to obtain a distinct mutant that ties against $\mathbf{x}$. Then $\mathbf x$ is an ESS if and only if $\mathbf d^\top\mathbf B\mathbf d<0$ for every $\mathbf d\in\mathcal D(T,J)$, with the condition vacuous when this set is empty. To see this, a mutant tying against $\mathbf x$ lies in the best-response face supported on $J$. Writing $\mathbf y=\mathbf x+t\mathbf d$ gives
\[
\mathbf y^\top\mathbf B\mathbf y-\mathbf x^\top\mathbf B\mathbf y
=t^2\mathbf d^\top\mathbf B\mathbf d.
\]
Here $\mathbf{d}^{\top}\mathbf{B}\mathbf{x}
= \sum_{j\in J} d_j v = 0$, since $(\mathbf{B}\mathbf{x})_j=v$ on $J$,
$d_j=0$ outside $J$, and the entries of $\mathbf{d}$ sum to zero. Conversely, for any $\mathbf{d}\in\mathcal{D}(T,J)$, a negative component
$d_j$ is possible only where $x_j>0$. Thus, for sufficiently small $t>0$,
$\mathbf{y}=\mathbf{x}+t\mathbf{d}$ is a mixed strategy distinct from
$\mathbf{x}$, and $\mathbf{d}^{\top}\mathbf{B}\mathbf{x}=0$ implies that
$\mathbf{y}$ ties against $\mathbf{x}$.
Compactness of the normalized direction set means that strict negativity is equivalent to a strictly negative maximum whenever the set is nonempty.

Fix a vector of positive margins $\boldsymbol{\epsilon}=(\epsilon_s,\epsilon_p,\epsilon_{\mathrm{stab}})$. Require $x_i\geq\epsilon_s$ on $T$ and enlarge the best-response set to
\[
J_{\epsilon_p}=\{j:v-(\mathbf B\mathbf x)_j<\epsilon_p\}.
\]
Our sufficient stability test requires
\[
\max_{\mathbf d\in\mathcal D(T,J_{\epsilon_p})}
\mathbf d^\top\mathbf B\mathbf d\leq-\epsilon_{\mathrm{stab}}.
\]
An empty direction set passes. Since $J\subseteq J_{\epsilon_p}$, the test covers every direction in $\mathcal D(T,J)$. Passing therefore implies that $\mathbf{x}$ is an ESS whenever the Nash conditions hold exactly. Some ESSs can fail these fixed positive margins. We denote the resulting restricted set of leader--follower pairs by $\mathcal R_{\boldsymbol\epsilon}$. All discrete subroutines share the payoff data and margin parameters. Their pseudocode assumes globally resolved optimization problems; an unresolved solve returns \texttt{UNRESOLVED}, not a certificate.

\begin{algorithm}[!ht]
\caption{\textsc{ComputeDiscreteOSESS}}
\label{alg:discrete-osess}
\begin{algorithmic}[1]
\REQUIRE Payoff tensor $\mathbf A_F$, leader payoff matrix $\mathbf A_L$
\REQUIRE Positive margins $\epsilon_s,\epsilon_p,\epsilon_{\mathrm{stab}}$
\ENSURE Best pair in $\mathcal R_{\boldsymbol\epsilon}$ if globally certified
\STATE $v^\star\leftarrow-\infty$; $(\boldsymbol\sigma^\star,\mathbf x^\star)\leftarrow(\varnothing,\varnothing)$
\FOR{all nonempty supports $T\subseteq\mathcal P$}
    \STATE $\mathcal C\leftarrow\varnothing$
    \REPEAT
        \STATE $(\texttt{status},\boldsymbol\sigma,\mathbf x)\leftarrow\textsc{SolveSESupport}(T,\mathcal C)$
        \IF{$\texttt{status}=\texttt{INFEASIBLE}$}
            \STATE \textbf{break}
        \ENDIF
        \IF{$\texttt{status}=\texttt{UNRESOLVED}$}
            \STATE \textbf{return} \texttt{UNRESOLVED}
        \ENDIF
        \STATE $(\texttt{test},\mathbf d)\leftarrow\textsc{CheckESS}(\boldsymbol\sigma,\mathbf x,T)$
        \IF{$\texttt{test}=\texttt{PASS}$}
            \STATE $v\leftarrow\boldsymbol\sigma^\top\mathbf A_L\mathbf x$
            \IF{$v>v^\star$}
                \STATE $v^\star\leftarrow v$; $(\boldsymbol\sigma^\star,\mathbf x^\star)\leftarrow(\boldsymbol\sigma,\mathbf x)$
            \ENDIF
            \STATE \textbf{break}
        \ELSIF{$\texttt{test}=\texttt{FAIL}$}
            \STATE $\mathcal C\leftarrow\mathcal C\cup\{\mathbf d\}$
        \ELSE
            \STATE \textbf{return} \texttt{UNRESOLVED}
        \ENDIF
    \UNTIL{\texttt{FALSE}}
\ENDFOR
\STATE \textbf{return} $(\boldsymbol\sigma^\star,\mathbf x^\star)$
\end{algorithmic}
\end{algorithm}

\begin{algorithm}[!ht]
\caption{\textsc{SolveSESupport}}
\label{alg:discrete-stack}
\begin{algorithmic}[1]
\REQUIRE Support $T$, generated directions $\mathcal C$; shared payoff data and margins
\ENSURE Globally optimal candidate for the current master problem,
\texttt{INFEASIBLE}, or \texttt{UNRESOLVED}
\STATE Solve globally, with $\mathbf B^{\boldsymbol\sigma}=\sum_\ell\sigma_\ell\mathbf A_F(\ell,\cdot,\cdot)$:
\[
\begin{aligned}
\max_{\boldsymbol\sigma,\mathbf x,v}\quad &\boldsymbol\sigma^\top\mathbf A_L\mathbf x\\
\text{s.t.}\quad &\boldsymbol\sigma\in\Delta(\mathcal M),\quad \mathbf 1^\top\mathbf x=1,\\
&x_i\geq\epsilon_s\ (i\in T),\quad x_j=0\ (j\notin T),\\
&(\mathbf B^{\boldsymbol\sigma}\mathbf x)_i=v\quad(i\in T),\\
&(\mathbf B^{\boldsymbol\sigma}\mathbf x)_j\leq v\quad(j\notin T),\\
&\textsc{DirectionCut}(\mathbf d;\boldsymbol\sigma,\mathbf x,v,T)
\quad(\mathbf d\in\mathcal C).
\end{aligned}
\]
\IF{global optimality is established}
    \STATE \textbf{return} $(\texttt{OPTIMAL},\boldsymbol\sigma,\mathbf x)$
\ELSIF{infeasibility is established}
    \STATE \textbf{return} $(\texttt{INFEASIBLE},\varnothing,\varnothing)$
\ELSE
    \STATE \textbf{return} $(\texttt{UNRESOLVED},\varnothing,\varnothing)$
\ENDIF
\end{algorithmic}
\end{algorithm}

\begin{algorithm}[ht]
\caption{\textsc{CheckESS}}
\label{alg:ess-check}
\begin{algorithmic}[1]
\REQUIRE Leader strategy $\boldsymbol\sigma$, symmetric Nash equilibrium $\mathbf x$ with support $T$
\REQUIRE Shared payoff data and positive margins $\epsilon_p,\epsilon_{\mathrm{stab}}$
\ENSURE \texttt{PASS}, a direction violating the margin test, or \texttt{UNRESOLVED}
\STATE $\mathbf B\leftarrow\mathbf B^{\boldsymbol\sigma}$; $v\leftarrow\mathbf x^\top\mathbf B\mathbf x$
\STATE $J_{\epsilon_p}\leftarrow\{j:v-(\mathbf B\mathbf x)_j<\epsilon_p\}$
\STATE Solve globally:
\[
\begin{aligned}
\rho=\max_{\mathbf d}\quad &\mathbf d^\top\mathbf B\mathbf d\\
\text{s.t.}\quad &\mathbf 1^\top\mathbf d=0,\quad\|\mathbf d\|_2^2=1,\\
&d_j\geq0\quad(j\notin T),\\
&d_j=0\quad(j\notin J_{\epsilon_p}).
\end{aligned}
\]
\IF{the direction set is certified empty}
    \STATE \textbf{return} $(\texttt{PASS},\varnothing)$
\ELSIF{a valid global upper bound is at most $-\epsilon_{\mathrm{stab}}$}
    \STATE \textbf{return} $(\texttt{PASS},\varnothing)$
\ELSIF{a feasible $\mathbf d$ has $\mathbf d^\top\mathbf B\mathbf d>-\epsilon_{\mathrm{stab}}$}
    \STATE \textbf{return} $(\texttt{FAIL},\mathbf d)$
\ELSE
    \STATE \textbf{return} $(\texttt{UNRESOLVED},\varnothing)$
\ENDIF
\end{algorithmic}
\end{algorithm}

For a generated direction $\mathbf d$ on fixed support $T$, set $S_{\mathbf d}=\{j\notin T:d_j>0\}$. The corresponding cut is the disjunction
\[
\mathbf d^\top\mathbf B^{\boldsymbol\sigma}\mathbf d\leq-\epsilon_{\mathrm{stab}}
\quad\text{or}\quad
v-(\mathbf B^{\boldsymbol\sigma}\mathbf x)_j\geq\epsilon_p
\ \text{for some }j\in S_{\mathbf d}.
\]
When $S_{\mathbf d}$ is empty, only the first alternative is available. To encode the general case, introduce binary variables $b_0$ and $b_j$ for $j\in S_{\mathbf d}$ and impose
\[
\begin{aligned}
&b_0+\sum_{j\in S_{\mathbf d}}b_j\geq1,\\
&\mathbf d^\top\mathbf B^{\boldsymbol\sigma}\mathbf d
 \leq-\epsilon_{\mathrm{stab}}+M_{\mathrm{cut}}(1-b_0),\\
&v-(\mathbf B^{\boldsymbol\sigma}\mathbf x)_j
 \geq\epsilon_p-M_{\mathrm{cut}}(1-b_j)
 \quad(j\in S_{\mathbf d}).
\end{aligned}
\]
For example, a valid choice is $M_{\mathrm{cut}}\geq\max\{\epsilon_p,nB_{\max}+\epsilon_{\mathrm{stab}}\}$, where $B_{\max}=\max_{\ell,i,j}|\mathbf A_F(\ell,i,j)|$: the Nash constraints make the payoff gaps nonnegative, and $|\mathbf d^\top\mathbf B^{\boldsymbol\sigma}\mathbf d|\leq nB_{\max}$ for a unit direction.

These cuts are valid for $\mathcal R_{\boldsymbol\epsilon}$. At any pair in $\mathcal R_{\boldsymbol{\epsilon}}$ with the same follower
support $T$, either some action in $S_{\mathbf d}$ lies outside
$J_{\epsilon_p}$, satisfying a payoff-gap alternative, or
$\mathbf d\in\mathcal D(T,J_{\epsilon_p})$ and the stability alternative
holds by definition of $\mathcal R_{\boldsymbol{\epsilon}}$.
The strict inequality defining $J_{\epsilon_p}$ and the weak inequality in the payoff-gap cut are complementary, including at the boundary. A failing direction therefore excludes the current candidate without excluding any member of $\mathcal R_{\boldsymbol\epsilon}$.

If the support-restricted master problems are solved to global optimality, the cut-generation loop terminates for each support, and the algorithm returns a nonempty candidate, then the returned candidate maximizes the leader payoff over $\mathcal R_{\boldsymbol\epsilon}$. Indeed, for a fixed support, the master problem with finitely many generated cuts is a relaxation of this restricted set. Once its optimizer passes the stability test, no higher-objective member of $\mathcal R_{\boldsymbol\epsilon}$ exists on that support. Taking the best certified candidate over all supports gives the result. This guarantee applies to the margin-restricted set rather than to all ESSs. For each fixed leader strategy $\boldsymbol{\sigma}$ and each ESS $\mathbf{x}$
of the induced follower game, the pair $(\boldsymbol{\sigma},\mathbf{x})$
belongs to $\mathcal R_{\boldsymbol{\epsilon}}$ for all sufficiently small
positive margins. Choose $\epsilon_s$ below its positive support masses, $\epsilon_p$ below its strictly positive payoff gaps, and $\epsilon_{\mathrm{stab}}$ below the magnitude of its negative directional maximum when the direction set is nonempty. The necessary margins are not known a priori, and finite termination is not asserted. Numerical implementations must also check Nash residuals and resolve ambiguous threshold comparisons before treating a result as a certificate.

Pseudocode for the main algorithm in continuous setting is given in Algorithm~\ref{alg:continuous-osess}. The algorithm follows a generate-and-certify structure: first a candidate satisfying the KKT conditions of the invasion maximization problems is computed, and then global bounds on mutant fitness are used to assess the no-invasion conditions. We assume that each fitness function is continuously differentiable in its scalar trait on $[0,1]$. A global trait maximizer exists on this interval and satisfies its box-constrained KKT conditions. The main subroutine Algorithm~\ref{alg:gen-osess} generates a candidate OSESS by solving a finite-dimensional KKT-based relaxation of the reformulation in Equation~\ref{eq:continuous-osess-reform}. The formulation globally optimizes the leader objective over solutions satisfying the KKT necessary optimality conditions for the invasion maximization problems, but such a solution is not guaranteed to satisfy the true global invasion constraints. We therefore use Algorithm~\ref{alg:certify} to bound the maximum mutant
fitness for each phenotype and compare those bounds with the invasion
threshold $\epsilon_{\mathrm{inv}}$, set to $10^{-3}$ in our experiment.
Passing this screen establishes neither exact nonpositive growth nor leader optimality by itself; ecological residuals and the leader objective bound must be assessed separately. For fixed $\mathbf{m}$ and $\mathbf{x}$, if $G_i$ is concave in the mutant
trait $u_i$, then its box-constrained KKT conditions are sufficient for a
global maximizer of mutant fitness for phenotype $i$.

In exact arithmetic, a globally optimal solution of this relaxation that satisfies all original ecological and invasion constraints is also globally optimal for Equation~\ref{eq:continuous-osess-reform}. Indeed, every feasible outcome of the original formulation extends to a feasible point of the relaxation by choosing a global mutant-fitness maximizer and its KKT multipliers for each phenotype. Thus the relaxation gives an upper bound on the leader optimum. If $\overline Q$ is a valid upper bound on the relaxation optimum and
$(\mathbf{m},\mathbf{u},\mathbf{x})$ is feasible for
Equation~\ref{eq:continuous-osess-reform}, then its leader payoff is within
$\overline Q-Q(\mathbf{m},\mathbf{u},\mathbf{x})$ of the optimal leader payoff.
If the candidate fails the invasion test, Algorithm~\ref{alg:continuous-osess}
returns \texttt{FAILURE} without searching for another candidate. An OSESS
may still exist. For a phenotype $i$ without an evolvable trait, we impose
$G_i(\mathbf{m},\mathbf{x})\leq 0$ directly in Algorithm~\ref{alg:gen-osess}
and evaluate this growth rate at the candidate in Algorithm~\ref{alg:certify}.

\begin{algorithm}[!ht]
\caption{\textsc{ComputeContinuousOSESS}}
\label{alg:continuous-osess}
\begin{algorithmic}[1]
\REQUIRE Leader objective $Q$, fitness functions $\{G_i\}_{i=1}^n$, feasible set $\mathcal M$
\REQUIRE Numerical invasion threshold $\epsilon_{\mathrm{inv}}$
\ENSURE Candidate passing the invasion screen, \texttt{FAILURE}, or \texttt{UNRESOLVED}
\STATE $(\texttt{status},\mathbf m^\star,\mathbf u^\star,\mathbf x^\star)\leftarrow\textsc{GenOSESS}(Q,\{G_i\},\mathcal M)$
\IF{$\texttt{status}=\texttt{INFEASIBLE}$}
    \STATE \textbf{return} \texttt{FAILURE}
\ELSIF{$\texttt{status}\neq\texttt{OPTIMAL}$}
    \STATE \textbf{return} \texttt{UNRESOLVED}
\ENDIF
\FOR{$i=1$ to $n$}
    \STATE $(L_i,U_i)\leftarrow\textsc{Certify}(i,\mathbf m^\star,\mathbf x^\star)$
\ENDFOR
\IF{$\max_i U_i\leq\epsilon_{\mathrm{inv}}$}
    \STATE \textbf{return} $(\mathbf m^\star,\mathbf u^\star,\mathbf x^\star)$
\ELSIF{$\max_i L_i>\epsilon_{\mathrm{inv}}$}
    \STATE \textbf{return} \texttt{FAILURE}
\ELSE
    \STATE \textbf{return} \texttt{UNRESOLVED}
\ENDIF
\end{algorithmic}
\end{algorithm}

\begin{algorithm}[!ht]
\caption{\textsc{GenOSESS}}
\label{alg:gen-osess}
\begin{algorithmic}[1]
\REQUIRE Leader objective $Q(\mathbf{m},\mathbf{u},\mathbf{x})$
\REQUIRE Fitness functions $\{G_i(u_i,\mathbf{m},\mathbf{x})\}_{i=1}^n$
\REQUIRE Feasible set $\mathcal{M}$
\ENSURE Globally optimal candidate, \texttt{INFEASIBLE}, or \texttt{UNRESOLVED}
\STATE Solve the following optimization problem:
\small
\[
\begin{aligned}
\max_{\mathbf{m},\mathbf{u},\mathbf{x},\mathbf{u}',\boldsymbol{\lambda}^L,\boldsymbol{\lambda}^U}
\quad
& Q(\mathbf{m},\mathbf{u},\mathbf{x}) \\
\text{s.t.}\quad
& x_i G_i(u_i,\mathbf{m},\mathbf{x}) = 0, \quad i=1,\ldots,n, \\
& \mathbf{m} \in \mathcal{M},\ \mathbf{x} \ge \mathbf{0},\
\mathbf{0} \le \mathbf{u} \le \mathbf{1},\ \mathbf{0} \le \mathbf{u}' \le \mathbf{1}, \\
& G_i(u'_i,\mathbf{m},\mathbf{x}) \le 0, \quad i=1,\ldots,n, \\
& \frac{\partial G_i}{\partial u_i}(u'_i,\mathbf{m},\mathbf{x})
+ \lambda_i^L - \lambda_i^U = 0,
\ i=1,\ldots,n, \\
& \lambda_i^L u'_i = 0,
\quad i=1,\ldots,n, \\
& \lambda_i^U (1-u'_i) = 0,
\quad i=1,\ldots,n, \\
& \lambda_i^L \ge 0,\;
\lambda_i^U \ge 0,
\quad i=1,\ldots,n .
\end{aligned}
\]
\normalsize

\IF{global optimality is established}
    \STATE \textbf{return} $(\texttt{OPTIMAL},\mathbf m^\star,\mathbf u^\star,\mathbf x^\star)$
\ELSIF{infeasibility is established}
    \STATE \textbf{return} $(\texttt{INFEASIBLE},\varnothing,\varnothing,\varnothing)$
\ELSE
    \STATE \textbf{return} $(\texttt{UNRESOLVED},\varnothing,\varnothing,\varnothing)$
\ENDIF

\end{algorithmic}
\end{algorithm}

\begin{algorithm}[!ht]
\caption{\textsc{Certify}}
\label{alg:certify}
\begin{algorithmic}[1]
\REQUIRE Phenotype index $i$ and fixed $(\mathbf m^\star,\mathbf x^\star)$
\ENSURE Bounds $L_i\leq G_i^{\max}\leq U_i$
\IF{phenotype $i$ has no evolvable trait}
    \STATE \textbf{return} $(G_i(\mathbf m^\star,\mathbf x^\star),G_i(\mathbf m^\star,\mathbf x^\star))$
\ENDIF
\STATE Apply global optimization to
\[
G_i^{\max}=\max_{0\leq v_i\leq1}G_i(v_i,\mathbf m^\star,\mathbf x^\star).
\]
\STATE Obtain $L_i$ by evaluating any feasible trait in the original fitness function
\STATE Obtain $U_i$ from a valid global objective bound
\STATE Use $L_i=-\infty$ or $U_i=+\infty$ when the corresponding bound is unavailable
\STATE \textbf{return} $(L_i,U_i)$
\end{algorithmic}
\end{algorithm}

Algorithms~\ref{alg:gen-osess}--\ref{alg:certify} use global-solve statuses and bounds in their mathematical sense. In a floating-point implementation, solver statuses and bounds are subject to numerical accuracy. In particular, an incumbent objective in the invasion maximization is a lower bound and cannot alone certify absence of invasion. We check that the original functions are defined at the candidate and
evaluate the residuals in the ecological equilibrium equations.
Section~\ref{se:exp} describes these checks and explains how we reconstruct
the reported OSESS candidate.

\begin{table}[!ht]
\centering
\caption{Illustrative positive margins for the discrete procedure and the invasion-screening threshold used in the continuous experiment.}
\label{tab:tols}
\begin{tabular}{|*{2}{c|}} \hline
Parameter &Value\\ \hline
$\epsilon_s$ & $10^{-4}$\\  \hline 
$\epsilon_p$ & $10^{-5}$ \\ \hline
$\epsilon_{\mathrm{stab}}$ & $10^{-5}$ \\ \hline
$\epsilon_{\mathrm{inv}}$ & $10^{-3}$ \\ \hline
\end{tabular}
\end{table}

\FloatBarrier
\section{Experiments}
\label{se:exp}
We apply our continuous OSESS algorithm to the evolutionary cancer game
studied by Kleshnina et al.~\cite{Kleshnina23:Game} and compare the resulting
treatment strategy, quality of life, and computation time with those obtained
using our previous SE algorithm~\cite{Ganzfried24:Computing}.
Kleshnina et al.~\cite{Kleshnina23:Game} formulate the Stackelberg equilibrium (SE) problem, where the leader optimizes quality of life $Q$ while anticipating fitness-maximizing traits at ecological equilibrium. 
The SE baseline explicitly enforces nonnegative population sizes, a constraint omitted from the original numerical implementation.
We compute a candidate OSESS by requiring ecological equilibrium and nonpositive growth for every invading phenotype or mutant trait. The two formulations differ in how they treat absent phenotypes. The SE baseline requires their per-capita growth rates to be zero, while the OSESS formulation also permits negative rates.
In both formulations, the treatment domain is $\mathcal{M}=\mathbb{R}_+^2$.
In the SE benchmark, a treatment is feasible only if the follower constraints
in Equation~\ref{eq:stack} admit a solution.
The equations below give the fitness functions $G_i$ and the quality-of-life
function $Q$ as implemented in the code accompanying Kleshnina et
al.~\cite{Kleshnina23:SEG}, which differs slightly from the published
model~\cite{Kleshnina23:Game}.
In this model, only cell types 1 and 2 have evolvable resistance traits. For each of these two types, we maximize mutant fitness over its resistance
trait in $[0,1]$, holding the treatment $\mathbf{m}$ and resident populations
$\mathbf{x}$ fixed. For type~0, the no-invasion condition reduces to
$G_0(\mathbf{m},\mathbf{x})\leq 0$, even when $x_0=0$.

\footnotesize
$$G_0 = r_{\mathrm{max}} \left(1 - \frac{\alpha_{00}x_0 + \alpha_{01}x_1 + \alpha_{02}x_2}{K}\right) - d - \frac{m_1}{k_1} - \frac{m_2}{k_2}$$
$$G_1 = r_{\mathrm{max}} e^{-g_1 u_1} \left(1 - \frac{\alpha_{10}x_0 + \alpha_{11}x_1 + \alpha_{12}x_2}{K}\right) - d - \frac{m_1}{b_1 u_1 + k_1} - \frac{m_2}{k_2}$$
$$G_2 = r_{\mathrm{max}} e^{-g_2 u_2} \left(1 - \frac{\alpha_{20}x_0 + \alpha_{21}x_1 + \alpha_{22}x_2}{K}\right) - d -  \frac{m_1}{k_1} - \frac{m_2}{b_2 u_2 + k_2}$$
$$Q = Q^{\mathrm{max}} - c \left(\frac{x_0 + x_1 + x_2}{K}\right)^2 - w_1 m^2_1 - w_2 m^2_2 - r_1 u^2_1 - r_2 u^2_2$$
\normalsize

The model has several parameters, whose interpretations are summarized in Table~\ref{ta:params}.
The code uses four parameters $a_0,a_1,a_2,a_3$ to specify the interaction
coefficients $\alpha_{ij}$ as follows.
\[
  \begin{bmatrix}
    \alpha_{00} &\alpha_{01} &\alpha_{02} \\
		\alpha_{10} &\alpha_{11} &\alpha_{12} \\
		\alpha_{20} &\alpha_{21} &\alpha_{22} \\
  \end{bmatrix} 
	= 
	 \begin{bmatrix}
    a_0 &a_1 &a_1 \\
		a_2 &a_0 &a_3 \\
		a_2 &a_3 &a_0 \\
  \end{bmatrix} 
\]

\begin{table}[!ht]
\centering
\caption{Interpretations of model parameters}
\label{ta:params}
\small
\begin{tabular}{|*{2}{c|}} \hline
Parameter &Interpretation\\ \hline
$r_{\mathrm{max}}$ &Max cell growth rate \\ 
$g_i$ &Cost of resistance strategy (cell type) $i$\\ 
$\alpha_{ij}$ &Interaction coefficient between cell types $i$ and $j$ \\ 
$K$ &Carrying capacity \\ 
$d$ &Natural death rate \\ 
$k_i$ &\makecell[c]{Innate resistance that may be\\ present before drug exposure} \\ 
$b_i$ &\makecell[c]{Benefit of the evolved resistance\\ 
	trait in reducing therapy efficacy} \\ 
$Q^{\mathrm{max}}$ &Quality of life of a healthy patient \\
$w_i$ &Toxicity of drug $i$ \\
$r_i$ &Effect of resistance rate of cell type $i$ \\
$c$ &\makecell[c]{Weight for impact of tumor burden\\
vs.\ drug toxicity/treatment-induced resistance rate} \\ \hline
\end{tabular}
\end{table}
\normalsize

For both the SE and OSESS computations, we used the parameter values
listed in Table~\ref{ta:param-values}, taken from the code accompanying
Kleshnina et al.~\cite{Kleshnina23:SEG}. Both approaches use Gurobi's nonconvex nonlinear solver version 14.0.0~\cite{Gurobi26:Gurobi}, with quadratic objectives and constraints and nonlinear exponential constraints. For the SE results in Table~\ref{ta:results}, we solved the formulation
from~\cite{Ganzfried24:Computing} using nonlinear exponential constraints.
The implementation reported in that study used a piecewise-linear
approximation of the exponential function.
We used Java 26.0.2 and an Intel Core i7-1065G7 processor with 8 threads and 16 GB of RAM. We set DualReductions to 0 for both approaches. The SE run uses the default feasibility tolerance $10^{-6}$ and relative optimality gap $10^{-4}$. For OSESS, we performed a validation run with feasibility and optimality tolerances of $10^{-9}$, relative optimality gap $10^{-5}$, NumericFocus set to 3, and a 300-second time limit. Both runs returned \texttt{OPTIMAL} status, subject to numerical tolerances, and the OSESS validation run finished before its time limit. We additionally check the original ecological and invasion constraints. The reported runtimes are for the global optimizations and exclude reconstruction and separate invasion checks. Since the formulations and numerical settings differ, these runtimes do not isolate the effect of imposing evolutionary stability.

\begin{table}[!ht]
\centering
\caption{Parameter values used in experiments}
\label{ta:param-values}
\begin{tabular}{|*{2}{c|}} \hline
Parameter &Value\\ \hline
$r_{\mathrm{max}}$ &0.45 \\ 
$g_1$ &0.5\\ 
$g_2$ &0.5\\ 
$a_0$ &1 \\
$a_1$ &0.15 \\
$a_2$ &0.9 \\
$a_3$ &0.9 \\
$K$ & 10,000 \\ 
$d$ & 0.01 \\ 
$k_1$ &5 \\ 
$k_2$ &5 \\ 
$b_1$ &10 \\ 
$b_2$ &10 \\ 
$Q^{\mathrm{max}}$ &1 \\
$w_1$ &0.5 \\
$w_2$ &0.2 \\
$r_1$ &0.4 \\
$r_2$ &0.4 \\
$c$ &0.5 \\ \hline
\end{tabular}
\end{table}

\begin{table}[!ht]
\centering
\caption{Experimental results for evolutionary cancer game}
\label{ta:results}
\begin{tabular}{|*{3}{c|}} \hline
&Stackelberg equilibrium &OSESS\\ \hline
$m_1$ &0.41058545 &0.42304257\\ \hline
$m_2$ &0.46760582 &0.47955704\\ \hline
$u_1$ &0.21414089 &0\\ \hline
$u_2$ &0.28525732 &0.28482720\\ \hline
$x_0$ &5731.5118 &5623.9363\\ \hline
$x_1$ &0.00013160 &0\\ \hline
$x_2$ &954.6240 &948.5840\\ \hline
$Q$ &0.59756546 &0.61608177\\ \hline
Runtime (seconds) &19.623 &271.735\\ \hline
\end{tabular}
\end{table}

Table~\ref{ta:results} shows that both treatment doses are higher in the
reconstructed OSESS candidate than in the SE solution. The corresponding
quality-of-life values are $0.61608$ and $0.59757$, respectively.
The SE solve took about 20 seconds, compared with about 4.5 minutes for the OSESS validation solve.
For both approaches we recompute $Q$ and $G_i$ directly from the returned values. The SE solution gives $G_0\approx0$, $G_1^{\max}\approx\num{2.47e-13}$ and $G_2^{\max}\approx\num{2.13e-13}$, so it satisfies the no-invasion conditions within numerical precision. The solver reports an upper bound of \(0.59756796\) on the optimal quality of life under the SE formulation.
The SE solution does not maximize quality of life under Equation~\ref{eq:continuous-osess-reform}, since the reconstructed OSESS candidate is feasible for that formulation and attains a higher value of $Q$.

For the OSESS relaxation, the solver returned a candidate with objective
value $0.61611476$ and reported an upper bound of $0.61612057$ on the
relaxation optimum. The maximum constraint violation reported for this
candidate was approximately $\num{1.78e-9}$. To obtain the feasible candidate
in Table~\ref{ta:results}, we set $x_1=0$ and solved the growth equations
for the two resident phenotypes exactly. Set $m_1=0.42304257$, $u_1=0$ and $u_2=0.2848272$, and define
\[
s=k_2+b_2u_2,\qquad C=d+\frac{m_1}{k_1},\qquad
m_2=\frac{g_2 C s^2}{b_2-g_2s}.
\]
With these parameter values, the resident growth equations
$G_0(\mathbf{m},\mathbf{x})=0$ and $G_2(u_2,\mathbf{m},\mathbf{x})=0$
are linear in $x_0$ and $x_2$. Solving them gives the population sizes
reported in Table~\ref{ta:results}. The choice of $m_2$ also ensures that
the type-2 mutant fitness has derivative zero at the resident trait $u_2$.

For each type $i\in\{1,2\}$, fixing the treatment and resident populations
gives a mutant fitness function of the form
\[
f_i(v)=A_i e^{-v/2}-\frac{m_i}{5+10v}-C_i,\qquad 0\leq v\leq 1,
\]
where $m_i\geq 0$ is the dose of drug $i$. If $A_i\leq 0$, then $f_i$
is nondecreasing. If $A_i>0$ and $m_i=0$, then $f_i$ is strictly
decreasing. In the remaining case, $A_i>0$ and $m_i>0$, and $f_i'(v)$
has the sign of
\[
\frac{10m_i e^{v/2}}{(5+10v)^2}-\frac{A_i}{2}.
\]
The first term is strictly decreasing on $[0,1]$ because its logarithmic
derivative is $1/2-20/(5+10v)<0$. Thus $f_i'$ can change sign only
from positive to negative. In all cases, every point satisfying the
box-constrained KKT conditions globally maximizes $f_i$ on $[0,1]$.
For these parameter values, the KKT relaxation is therefore exact in
exact arithmetic, although we still check the feasibility of the numerical
solution. The same argument applies to the SE baseline: each resident
trait maximizes its phenotype's fitness, and Equation~\ref{eq:stack}
requires $G_0=G_1=G_2=0$. Thus every feasible SE outcome satisfies the
no-invasion conditions. The difference between the two feasible sets in
this experiment arises because the SE baseline requires zero per-capita
growth even for absent phenotypes.

At both reported outcomes, $A_i>0$ and $m_i>0$ for $i=1,2$. As an
independent check, we evaluated each mutant fitness function at both
endpoints and at any stationary point in $(0,1)$. For the reconstructed
OSESS candidate, $G_0=0$ and the maximum type-2 mutant fitness is zero
by construction. The numerical check gives a maximum type-1 mutant
fitness of approximately $-\num{1.00e-8}$.
To verify the sign of the type-1 maximum, we used rational interval
arithmetic with Taylor bounds for the exponential and an interval
enclosing the unique stationary point of the type-1 mutant fitness
function. The resulting bound, $G_1^{\max}<-\num{1.0049e-8}$, together
with the equalities above, establishes that the unrounded candidate is
feasible for Equation~\ref{eq:continuous-osess-reform}.
Since type~1 is absent in the reconstructed candidate, setting its
resident trait to $u_1=0$ maximizes $Q$ with all other variables fixed.
The invasion check still considers every possible type-1 mutant trait
in $[0,1]$. The penalty $r_1u_1^2$ decreases from approximately
$0.01834$ at the SE solution to zero at the reconstructed candidate,
accounting for about $99\%$ of the $0.01852$ increase in $Q$. 

The reconstructed candidate has quality of life $Q=0.6160817687$. Relative to this value, the gap to the solver's numerical
upper bound of $0.61612057$ is approximately $0.0063\%$. The accuracy of
this gap depends on the numerical accuracy of the upper bound. We used
unrounded values for all checks and rounded the entries in
Table~\ref{ta:results} for presentation. The reported OSESS results use
the validation settings and reconstruction described above. The default
program entry point retains the original solver settings.

\section{Conclusion}
\label{se:conc}
Stackelberg evolutionary games model biological interactions such as cancer therapy, fishery management, and pest control~\cite{Kleshnina23:Game}. In these games, a rational human leader first selects a strategy, and evolutionary followers then respond conditional on the leader's strategy, possibly subject to ecological population equilibrium constraints. The prior formulations optimize the leader objective while anticipating ecological equilibrium and fitness-maximizing trait responses. In our SESS framework, the leader chooses a strategy to maximize their objective, anticipating how a follower response will be selected from the set of stable outcomes. In discrete games, the ESS requirement is stronger than the symmetric Nash equilibrium conditions. In the continuous model studied here, we require ecological equilibrium and nonpositive invasion fitness for every phenotype and mutant trait, including absent phenotypes. This is a first-order no-invasion criterion. Convergence and asymptotic stability require additional analysis.

For each leader strategy, optimistic SESS selects a stable follower response that maximizes the leader’s payoff, while pessimistic SESS selects one that minimizes it. Thus, under optimistic selection, we can maximize the leader’s payoff jointly over the leader strategy and an admissible stable follower response. Our discrete algorithm optimizes over a margin-restricted subset of ESS outcomes, with conditional global guarantees for that subset. Our continuous procedure combines a KKT relaxation with checks of the original ecological and invasion constraints. In the discrete model, the follower’s mixed strategy can represent the frequencies of phenotypes in a population. For example, a large population of cancer cells can be described by the fraction of cells of each phenotype.

We implemented our OSESS algorithm on a previously studied evolutionary cancer game model and evaluated its runtime and solution quality. The SE solution satisfies the no-invasion conditions within numerical precision. Under the OSESS formulation, which permits negative per-capita growth for an absent phenotype, we found a feasible outcome with higher doses of both drugs and higher quality of life. At this outcome, type~1 is absent, and even its most favorable resistance
trait gives it a negative growth rate if it enters the resident population.
The SE benchmark excludes this outcome because it requires zero per-capita
growth even for an absent phenotype. In the discrete example, the leader’s preferred symmetric Nash equilibrium fails the ESS condition, so optimistic SE and OSESS select different follower responses in the same game. For larger games, it may be useful to compute an optimistic SE first and then test its follower response for evolutionary stability. If the SE formulation uses the same ecological constraints and includes every pair admissible under OSESS, then any globally optimal SE pair that satisfies the stability conditions exactly is also an OSESS.

\bibliographystyle{plain}
\bibliography{C://FromBackup/Research/refs/dairefs}

\end{document}